\documentclass{amsart}

\newtheorem{thm}{Theorem}

\begin{document}
\title[Hidden Markov Chains with rare transitions]{Entropy Rate for
Hidden Markov Chains with rare transitions} \date{\today}
\author{Yuval Peres}
\address{Microsoft Research, One Microsoft Way, Redmond, WA 98052,
  USA}
\email{peres(a)microsoft.com}
\author{Anthony Quas}
\address{Department of Mathematics and Statistics, University of
  Victoria, Victoria, BC V8W 3R4, CANADA}
\email{aquas(a)uvic.ca}
\thanks{AQ's research was supported by NSERC; The authors thank BIRS
  where this research was conducted.}

\begin{abstract}
We consider Hidden Markov Chains obtained by passing a Markov Chain with
rare transitions through a noisy memoryless channel. We obtain asymptotic
estimates for the entropy of the resulting Hidden Markov Chain as the
transition rate is reduced to zero.
\end{abstract}

\maketitle

Let $(X_n)$ be a Markov chain with finite state space $S$ and transition
matrix $P(p)$ and let $(Y_n)$ be the Hidden Markov chain observed by
passing $(X_n)$ through a homogeneous noisy memoryless channel (i.e. $Y$
takes values in a set $T$, and there exists a matrix $Q$ such that
$\mathbb{P}(Y_n=j|X_n=i,X_{-\infty}^{n-1},X_{n+1}^{\infty}
,Y_{-\infty}^{n-1},Y_{n+1}^\infty) =Q_{ij}$). We make the additional
assumption on the channel that the rows of $Q$ are distinct. In this case
we call the channel \emph{statistically distinguishing}.

We assume that $P(p)$ is of the form $I+pA$ where $A$ is a matrix with
negative entries on the diagonal, non-negative entries in the
off-diagonal terms and zero row sums. We further assume that for small
positive $p$, the Markov chain with transition matrix $P(p)$ is
irreducible. Notice that for Markov chains of this form, the invariant
distribution $(\pi_i)_{i\in S}$ does not depend on $p$. In this case, we
say that for small positive values of $p$, the Markov chain is in a
\emph{rare transition regime}.

We will adopt the convention that $H$ is used to denote the entropy of a
finite partition, whereas $h$ is used to denote the entropy of a process
(the \emph{entropy rate} in information theory terminology). Given an
irreducible Markov chain with transition matrix $P$, we let $h(P)$ be the
entropy of the Markov chain (i.e. $h(P)=-\sum_{i,j} \pi_iP_{ij}\log
P_{ij}$ where $\pi_i$ is the (unique) invariant distribution of the
Markov chain and as usual we adopt the convention that $0\log 0=0$). We
also let $H_{\text{chan}}(i)$ be the entropy of the output of the channel
when the input symbol is $i$ (i.e. $H_{\text{chan}}(i)=-\sum_{j\in
T}Q_{ij}\log Q_{ij}$). Let $h(Y)$ denote the entropy of $Y$ (i.e.
$h(Y)=-\lim_{N\to\infty} \frac1N \sum_{w\in T^N}\mathbb{P}(Y_1^N=w)\log
\mathbb{P}(Y_1^N=w)$).

\begin{thm}\label{thm:main}
Consider the Hidden Markov Chain $(Y_n)$ obtained by observing a
Markov chain with irreducible transition matrix $P(p)=I+Ap$ through
a statistically distinguishing channel with transition matrix $Q$.
Then there exists a constant $C>0$ such that for all small $p>0$,
\begin{equation}\label{eq:mainbound}
  h(P(p))+\sum_i\pi_i H_{\text{chan}}(i)-Cp \le h(Y)
  \le h(P(p))+\sum_i\pi_i H_{\text{chan}}(i),
\end{equation}
where $(\pi_i)_{i\in S}$ is the invariant distribution of $P(p)$.

If in addition the channel has the property that there exist $i,i'$
and $j$ such that $P_{ii'}>0$, $Q_{ij}>0$ and $Q_{i'j}>0$, then there
exists a constant $c>0$ such that
\begin{equation}\label{eq:lowerbound}
  h(Y)\le h(P(p))+\sum_i\pi_i H_\text{chan}(i)-cp.
\end{equation}
\end{thm}

The entropy rate in the rare transition regime was considered
previously in the special case of a 0--1 valued Markov Chain with
transition matrix $
P(p)=\left(\begin{smallmatrix}1-p&p\\p&1-p\end{smallmatrix}\right) $
and where the channel was the binary symmetric channel with crossover
probability $\epsilon$ (i.e.  $ Q=\left(\begin{smallmatrix}
    1-\epsilon&\epsilon\\\epsilon&1-\epsilon\end{smallmatrix} \right)
$). It is convenient to introduce the notation $g(p)=-p\log
p-(1-p)\log(1-p)$.  In \cite{NOW}, Nair, Ordentlich and Weissman
proved that $g(\epsilon)-(1-2\epsilon)^2p\log p/(1-\epsilon)\le
h(Y)\le g(p)+g(\epsilon)$. For comparison, with our result, this is
essentially of the form $g(\epsilon)+a(\epsilon)g(p)\le h(Y)\le
g(p)+g(\epsilon)$ where $a(\epsilon)<1$ but $a(\epsilon)\to 1$ as
$\epsilon\to 0$ (i.e.  $h(Y)=g(p)+g(\epsilon)-O(p\log p)$). A second
paper due to Chigansky \cite{Chigansky} shows that
$g(\epsilon)+b(\epsilon)g(p)\le h(Y)$ for a function $b(\epsilon)<1$
satisfying $b(\epsilon)\to 1$ as $\epsilon\to 1/2$ (again giving an
$O(p\log p)$ error). Our result states in this case that there exist
$C>c>0$ such that $g(p)+g(\epsilon)-Cp\le h(Y)\le g(p)+g(\epsilon)-cp$
(i.e. $h(Y)=g(p)+g(\epsilon)-\Theta(p)$).

We note that as part of the proof we attempt a reconstruction of $(X_n)$
from the observed data $(Y_n)$. In our case, the reconstruction of the
$n$th symbol of $X_n$ depended on past and future values of $Y_m$. A
related but harder problem of filtering is to try to reconstruct $X_n$
given only $Y_1^n$. This problem was addressed in essentially the same
scenario by Khasminskii and Zeitouni \cite{KhasminskiiZeitouni}, where
they gave a lower bound for the asymptotic reconstruction error of the
form $Cp|\log p|$ for an explicit constant $C$ (i.e. for an arbitrary
reconstruction scheme, the probability of wrongly guessing $X_n$ is
bounded below in the limit as $n\to\infty$ by $Cp|\log p|$). Our scheme
shows that if one is allowed to use future as well as past observations
then the asymptotic reconstruction error is $O(p)$. This was previously
observed by Shue, Anderson and DeBruyne in \cite{SAB} who used a similar
scheme to ours.

Before giving the proof of the theorem, we discuss the strategy. We start
from the equality
\begin{equation}\label{eq:basic}
  h(X)+h(Y|X)=h(X,Y)=h(Y)+h(X|Y).
\end{equation}
Since $h(X)$ and $h(Y|X)$ are known to be $h(P(p))$ and $\sum_i \pi_i
H_\text{chan}(i)$, the estimates for the entropy of $Y$ are obtained by
estimating $h(X|Y)$. The inequality \eqref{eq:mainbound} is equivalent to
showing that $0\le h(X|Y)\le Cp$ for some $C>0$. The lower bound here is
trivial, whereas the main part of the proof is the upper bound for
$h(X|Y)$ (giving a lower bound for $h(Y)$). The second part of the proof,
showing \eqref{eq:lowerbound} lowering the upper bound for $h(Y)$ under
additional conditions, is proved by showing $h(X|Y)\ge cp$ for some
$c>0$.

We explain briefly the underlying idea of the upper bound $h(X|Y)=O(p)$.
Since the transitions in the $(X_n)$ sequence are rare, given a
realization of $(Y_n)$, the $Y_n$ values allow one to guess (using the
statistical-distinguishing property) the $X_n$ values from which the
$Y_n$ values are obtained. This provides for an accurate reconstruction
except that where there is a transition in the $X_n$'s there is some
uncertainty as to its location as estimated using the $Y_n$'s. It turns
out that by using maximum likelihood estimation, the transition locations
may be pinpointed up to an error with exponentially small tail. Since the
transitions occur with rate $p$, there is an $O(p)$ entropy error in
reconstructing $(X_n)$ from $(Y_n)$.

We make use of a number of notational conventions, some standard and
others less so. Firstly we shall write denote events by set notation so
that $\{X_0=X_2\}$ denotes the event that the random variables $X_0$ and
$X_2$ agree. We make extensive use of relative entropy. For two
partitions $\mathcal P$ and $\mathcal Q$, the relative entropy is defined
by $H(\mathcal Q|\mathcal P)=H(\mathcal P\vee\mathcal Q)-H(\mathcal P)$.
When conditioning, we shall not distinguish between random variables and
the partitions and $\sigma$-algebras that they induce (so that for
example $H(X_0^{N-1})$ is $-\sum_{w\in
S^N}\mathbb{P}(X_0^{N-1}=w)\log\mathbb{P}(X_0^{N-1}=w)$ and $H(X_0|Y)$ is
the conditional entropy of $X_0$ relative to the $\sigma$-algebra
generated by $\{Y_n\colon n\in\mathbb{Z}\}$). On the other hand if $A$ is
an event, we use $H(\mathcal P|A)$ to mean the entropy of the partition
with respect to the conditional measure $\mathbb{P}_A(B)=\mathbb{P}(A\cap
B)/\mathbb{P}(A)$. For jointly stationary processes
$(X_n)_{n\in\mathbb{Z}}$ and $(Y_n)_{n\in\mathbb{Z}}$, the relative
entropy of the processes is given by
$h(Y|X)=h((X_n,Y_n)_{n\in\mathbb{Z}})-h((X_n)_{n\in\mathbb{Z}})=
\lim_{N\to\infty}(1/N)(H(X_0^{N-1}\vee
Y_0^{N-1})-H(X_0^{N-1}))=\lim_{N\to\infty}(1/N)H(Y_0^{N-1}|X_{-\infty}^\infty)
=H(Y_0|X_{-\infty}^\infty,Y_{-\infty}^{-1})$.

Given a measurable partition $\mathcal Q$ of the space, an event $A$ and
a $\sigma$-algebra $\mathcal F$ we will write $H(\mathcal Q|\mathcal
F|A)$ for the entropy of $\mathcal Q$ relative to $\mathcal F$ with
respect to $\mathbb{P}_A$. In the case where $A$ is $\mathcal
F$-measurable (as will always be the case in what follows), we have
\begin{equation*}
  H(\mathcal Q|\mathcal F|A)=\int\left(-\sum_{B\in\mathcal
  Q}\mathbb{P}(B|\mathcal F)\log\mathbb{P}(B|\mathcal F)\right)\,d\mathbb{P}_A.
\end{equation*}

 If $A_1,\ldots,A_k$ form an $\mathcal
F$-measurable partition of the space, then we have the following
equality:
\begin{equation}\label{eq:entpart}
  H(\mathcal Q|\mathcal F)=\sum_{j=1}^k \mathbb{P}(A_j)H(\mathcal Q|\mathcal
  F|A_j).
\end{equation}

\begin{proof}[Proof of Theorem \ref{thm:main}]

Note that $((X_n,Y_n))_{n\in\mathbb{Z}}$ forms a Markov chain with
transition matrix $\bar P$ given by $\bar
P_{(i,j),(i',j')}=P_{ii'}Q_{i'j'}$ and invariant distribution $\bar
\pi_{(i,j)}=\pi_iQ_{ij}$. The standard formula for the entropy of a
Markov chain then gives $h(X,Y)=h(P(p))+\sum_i\pi_iH_{\text{chan}}(i)$.
Since $h(X,Y)=h(Y)+h(X|Y)$, one obtains
\begin{equation}\label{eq:basic2}
  h(Y)=h(X,Y)-h(X|Y)= h(P(p))+\sum_i\pi_iH_\text{chan}(i)-h(X|Y).
\end{equation}

This establishes the upper bound in the first part of the theorem.

We now establish the lower bound. We are aiming to show $h(X|Y)=O(p)$
(for which it suffices to show $H(X_0^{L-1}|Y)=O(Lp)$ for some $L$).
Setting $L=|\log p|^4$ and letting $\mathcal P$ be a suitable partition,
we estimate $H(X_0^{L-1}|Y,\mathcal P)$ and use the inequality

\begin{equation}\label{eq:entpert}
H(X_0^{L-1}|Y)\le H(X_0^{L-1}|Y,\mathcal P)+H(\mathcal P).
\end{equation}

We define the partition $\mathcal P$ as follows: Set $K=|\log p|^2$ and
let $\mathcal P=\{ E_\text{m}, E_\text{b},E_\text{g1},E_\text{g2}\}$.
Here $E_\text{m}$ (for many) is the event that there are at least two
transitions in $X_0^{L-1}$, $E_\text{b}$ (for boundary) is the event that
there is exactly one transition and that it takes place within a distance
$K$ of the boundary of the block and finally $E_\text{g}$ (for good) is
the event that there is at most one transition and if it takes place,
then it occurs at a distance at least $K$ from the boundary of the block.
This will later be subdivided into $E_\text{g1}$ and $E_\text{g2}$.

If $E_\text{m}$ holds then we bound the entropy contribution by the
entropy of the equidistributed case whereas if $E_\text{b}$ holds, there
are $2K|S|(|S|-1)=O(K)$ possible values of $X_0^{L-1}$. This yields the
following estimates:
\begin{align}
\mathbb{P}(E_\text{m})&=O(p^2L^2)=o(p) \label{eq:Emprob}\\
H(X_0^{L-1}|E_\text{m})&\le L\log |S|  \label{eq:Ement}  \\
\mathbb{P}(E_\text{b})&=O(pK)\label{eq:Ebprob}  \\
H(X_0^{L-1}|E_\text{b})&=O(\log K)\label{eq:Ebent}.
\end{align}

It follows that $\mathbb{P}(E_g)=1-O(pK)$. Given that the event
$E_\text{g}$ holds, the sequence $X_0^{L-1}$ belongs to $B=\{a^L\colon
a\in S\}\cup \{a^ib^{L-i}\colon a,b\in S, K\le i\le L-K\}$.

Given a sequence $u\in B$,  the log-likelihood of $u$ being the input
sequence yielding the output $Y_0^{L-1}$ is
$L_u(Y_0^{L-1})=\sum_{i=0}^{L-1}\log Q_{u_iY_i}$. We define $Z_0^{L-1}$
to be the sequence in $B$ for which $L_Z(Y_0^{L-1})$ is maximized
(breaking ties lexicographically if necessary). We will then show using
large deviation methods that when $E_\text{g}$ holds, $Z_0^{L-1}$ is a
good reconstruction of $X_0^{L-1}$ with small error.

We calculate for $u,v\in B$,
\begin{align*}
&\mathbb{P}\left(L_v(Y_0^{L-1})\ge
L_u(Y_0^{L-1})|X_0^{L-1}=u\right)\\=&\mathbb{P}\left(\sum_{i=0}^{L-1}\log
(Q_{v_iY_i}/Q_{u_iY_i})\ge 0|X_{0}^{L-1}=u\right)\\
=&\mathbb{P}\left(\sum_{i\in \Delta} \log (Q_{v_iY_i}/Q_{u_iY_i})\ge
0|X_{0}^{L-1}=u\right),
\end{align*}
where $\Delta=\{i\colon u_i\ne v_i\}$. For each $i\in \Delta$, given
that $X_0^{L-1}=u$, we have that $\log(Q_{v_iY_i}/Q_{u_iY_i})$ is an
independent random variable taking the value $\log
(Q_{v_ij}/Q_{u_ij})$ with probability $Q_{u_ij}$.

It is well known (and easy to verify using elementary calculus) that for
a given probability distribution $\pi$ on a set $T$, the probability
distribution $\sigma$ maximizing $\sum_{j\in T}\pi_j\log
(\sigma_j/\pi_j)$ is $\sigma=\pi$ (for which the maximum is 0).
Accordingly we see that given that $X_0^{L-1}=u$,
$L_v(Y_0^{L-1})-L_u(Y_0^{L-1})$ is the sum of $|\Delta|$ random
variables, each having one of $|S|(|S|-1)$ distributions, each with
negative expectation. It follows from Hoeffding's Inequality
\cite{Hoeffding} that there exist $C>0$ and $\eta<1$ independent of $p$
such that $\mathbb{P}(L_v(Y_0^{L-1})\ge L_u(Y_0^{L-1})|X_0^{L-1}=u)\le
C\eta^{|\Delta|}$ .

We deduce that for $u,v\in B$
\begin{equation}\label{eq:LD}
\mathbb{P}(Z_0^{L-1}=v|X_0^{L-1}=u)\le C\eta^{\delta(u,v)},
\end{equation}
where $\delta(u,v)$ is the number of places in which $u$ and $v$ differ.

We split $E_g$ into two subsets:
\begin{align*}
E_\text{g1}&=E_g\cap\{\delta(X_0^{L-1},Z_0^{L-1})<K\};\text{ and }\\
E_\text{g2}&=E_g\cap\{\delta(X_0^{L-1},Z_0^{L-1})\ge K\}.
\end{align*}
Since there are less than $|S|^2L$ elements in $B$, we see using
\eqref{eq:LD} and recalling that $K=|\log p|^2$ that
\begin{align}\label{eq:Eg2prob}
\mathbb{P}(E_\text{g2})&\le |S|^2LC\eta^K=o(p)\\
H(X_0^{L-1}|E_\text{g2})&\le \log(|S|^2L).\label{eq:Eg2ent}
\end{align}

Combining \eqref{eq:Eg2prob} with \eqref{eq:Ebprob} and \eqref{eq:Emprob}
we see that $\mathbb{P}(E_\text{g1})=1-O(pK)$. We then obtain
\begin{equation}\label{eq:entP}
H(\mathcal P)=O(pK\log(pK))=o(pL).
\end{equation}

Conditioned on being in $E_\text{g1}$, if $Z_0=Z_{L-1}$ then
$X_0^{L-1}=Z_0^{L-1}$ so we have
\begin{equation}\label{eq:notrans}
H(X_0^{L-1}|Z_0^{L-1}\vee\mathcal P|E_\text{g1}\cap \{Z_0=Z_{L-1}\})=0.
\end{equation}

Given that $E_\text{g1}$ holds, if $X_0^{L-1}=a^ib^{L-i}$ then
$Z_0^{L-1}$ must be of the form $a^jb^{L-j}$ for some $j$ satisfying
$-K<j-i<K$. Denote this difference $j-i$ by the random variable $N$.  We
have
\begin{align*}
&H(X_0^{L-1}|Y_0^{L-1}\vee \mathcal P|E_\text{g1}\cap \{Z_0\ne
Z_{L-1}\})\\
\le &H(X_0^{L-1}|Z_0^{L-1}\vee \mathcal P|E_\text{g1}\cap \{Z_0\ne
Z_{L-1}\})\\
=&H(N|Z_0^{L-1}\vee\mathcal P|E_\text{g1}\cap\{Z_0\ne Z_{L-1}\})\\
\le& H(N|E_\text{g1}\cap\{Z_0\ne Z_{L-1}\}).
\end{align*}
where the first inequality follows because $Z_0^{L-1}$ is determined by
$Y_0^{L-1}$ so the partition generated by $Y_0^{L-1}$ is finer than that
generated by $Z_0^{L-1}$; and the equality follows because given
$Z_0^{L-1}$ and conditioned on being in $E_\text{g1}$, knowing $N$ is
sufficient to reconstruct $X_0^{L-1}$ so the partition generated by $N$
is the same as the partition generated by $X_0^{L-1}$.

Since $E_\text{g1}\cap\{Z_0\ne Z_{L-1}\}=E_\text{g1}\cap \{X_0\ne
X_{L-1}\}$, we have for $|k|<K$, $\mathbb{P}(N=k|E_\text{g1}\cap\{Z_0\ne
Z_{L-1}\})=\mathbb{P}(N=k|E_\text{g1}\cap\{X_0\ne X_{L-1}\})$. From
\eqref{eq:LD} this is bounded above by $C\eta^{|k|}$. Since a
distribution with these bounds has entropy bounded above independently of
$p$, it follows from this that $H(N|E_\text{g1}\cap\{Z_0\ne
Z_{L-1}\})=O(1)$ and hence that
\begin{equation}\label{eq:htrans}
H(X_0^{L-1}|Y_0^{L-1}\vee\mathcal P|E_\text{g1}\cap\{Z_0\ne
Z_{L-1}\})=O(1).
\end{equation}
Finally we have $\mathbb{P}(E_\text{g1}\cap \{Z_0\ne Z_{L-1}\})=O(pL)$

We now have $H(X_0^{L-1}|Y_0^{L-1})\le H(X_0^{L-1}|Y_0^{L-1}\vee\mathcal
P)+H(\mathcal P)$. We estimate the right side using \eqref{eq:entpart},
splitting the space up into the sets $E_{\text b}$, $E_{\text m}$,
$E_{\text{g2}}$, $E_{\text{g1}}\cap\{Z_0=Z_{L-1}\}$ and
$E_{\text{g1}}\cap\{Z_0\ne Z_{L-1}\}$. All of these sets are
$Y_0^{L-1}\vee\mathcal P$ measurable. Calculating the contribution to the
entropy from each of the sets, each part contributes at most $O(pL)$
yielding the estimate $H(X_0^{L-1}|Y_0^{L-1})=O(pL)$, so that
$h(X|Y)=O(p)$ as required. This completes the first part of the proof.

For the second part of the proof, suppose that the additional properties
are satisfied (the existence of $i,i'$ and $j$ such that $P_{ii'}>0$,
$Q_{ij}>0$ and $Q_{i'j}>0$). We need to show that $h(X|Y)\ge cp$ for some
$c>0$ or equivalently that $H(X_0|Y,X_{-\infty}^{-1})\ge cp$. In fact, we
show the stronger statement: $H(X_0|Y,(X_n)_{n\ne 0})\ge cp$. Let $A$ be
the event that $X_{-1}=i$ and $X_1=i'$ and $Y_0=j$. We now estimate
$H(X_0|Y,(X_n)_{n\ne 0}|A)$.

For $x\in A$, we have
\begin{align*}
\mathbb{P}(X_0=i|Y,(X_n)_{n\ne 0})(x)&=\frac{P_{ii}P_{ii'}Q_{ij}}
{P_{ii}P_{ii'}Q_{ij}+P_{ii'}P_{i'i'}Q_{i'j}+
\sum_{k\not\in\{i,i'\}} P_{ik}P_{ki'}Q_{kj}}\\
\mathbb{P}(X_0=i'|Y,(X_n)_{n\ne 0})(x)
&=\frac{P_{ii'}P_{i'i'}Q_{i'j}}{P_{ii}P_{ii'}Q_{ij}+
P_{ii'}P_{i'i'}Q_{i'j}+\sum_{k\not\in\{i,i'\}} P_{ik}P_{ki'}Q_{kj}}.
\end{align*}

As $p\to 0$, we have $\mathbb{P}(X_0=i|Y,(X_n)_{n\ne 0})(x)\to
Q_{ij}/(Q_{ij}+Q_{i'j})$ and $\mathbb{P}(X_0=i'|Y,(X_n)_{n\ne 0})(x)\to
Q_{i'j}/(Q_{ij}+Q_{i'j})$. From this we see that $H(X_0|Y,(X_n)_{n\ne
0})|A)$ converges to a non-zero constant as $p\to 0$. Since $A$ has
probability $\Omega(p)$, applying \eqref{eq:entpart} we obtain the lower
bound $h(X|Y)\ge cp$. From this we deduce the claimed upper bound for
$h(Y)$:
\begin{equation*}
  h(Y)\le h(X)+\sum_i \pi_i H_\text{chan}(i)-cp.
\end{equation*}

In this case we therefore have $h(Y)=h(X)+\sum_i\pi_i
H_\text{chan}(i)+\Theta(p)$. This completes the proof of the theorem.

\end{proof}

\bibliographystyle{amsplain}

\providecommand{\bysame}{\leavevmode\hbox to3em{\hrulefill}\thinspace}
\providecommand{\MR}{\relax\ifhmode\unskip\space\fi MR }
\providecommand{\MRhref}[2]{%
  \href{http://www.ams.org/mathscinet-getitem?mr=#1}{#2}
}
\providecommand{\href}[2]{#2}

\end{document}